\documentclass[11pt]{amsart}
\usepackage{graphicx}
\newtheorem{thm}{Theorem}[section]

\theoremstyle{definition}

\theoremstyle{remark}

\numberwithin{equation}{section}
\def\bc{{\mathbb C}}

\def\bm{{\mathbb M}}
\def\bn{{\mathbb N}}

\def\br{{\mathbb R}}

\def\a{\alpha}
\def\b{\beta}
\def\g{\gamma}  
  \def\D{\Delta}

\def\l{\lambda} 

\def\m{\mu}

\def\s{\sigma} 
\def\t{\tau}

\def\tr{\mathop{\rm Tr}}
\def\id{{\bf 1}\!\!{\rm I}}

\def\D{\Delta}

\def\wb{{\mathbf{w}}}

\def\a{\alpha}
\begin{document}
\title[bistochastic Kadison-Schwartz operators]
{On description of bistochastic Kadison-Schwarz operators on $\bm_2(\mathbb{C})$}

\author{Farrukh Mukhamedov}
\address{Farrukh Mukhamedov\\
 Department of Computational \& Theoretical Sciences\\
Faculty of Science, International Islamic University Malaysia\\
P.O. Box, 141, 25710, Kuantan\\
Pahang, Malaysia} \email{{\tt far75m@yandex.ru}, {\tt
farrukh\_m@iiu.edu.my}}

\author{Abduaziz Abduganiev}
\address{Abduaziz Abduganiev\\
 Department of Computational \& Theoretical Sciences\\
Faculty of Science, International Islamic University Malaysia\\
P.O. Box, 141, 25710, Kuantan\\
Pahang, Malaysia} \email{{\tt azizi85@yandex.ru}}
%\thanks{}%
\begin{abstract}

In this paper we describe bistochastic Kadison-Schawrz operators on $M_2(\mathbb{C})$.
Such a description allows us to find
positive, but not Kadison-Schwarz operators. Moreover, by means of that characterization we construct  Kadison-Schawrz operators, which are not completely positive.

\end{abstract}

\maketitle
% ----------------------------------------------------------------

\section{Introduction}

It is known that the theory of quantum dynamical systems provides a convenient
mathematical description of irreversible dynamics of an open quantum system (see\cite{BR})
investigation of various properties of such dynamical systems have had a
considerable growth. In a quantum setting, the matter is more complicated than in the
classical case. Some differences between classical and quantum situations are pointed out in
\cite{NC}. This motivates an interest to study dynamics of quantum
systems (see\cite{NC}). One of the main objects of this theory is mapping
(or channel) defined on matrix algebras. One of the main constraints to such a mapping is
positivity and complete positivity.  There are many papers devoted to this problem (see
for example \cite{Choi,MM1,St,W}). In the literature the most
tractable maps, the completely positive ones, have proved to be of
great importance in the structure theory of C$^*$-algebras. However,
general positive (order-preserving) linear maps are very
intractable\cite{MM1,Ma1, Ma2}. It is therefore of interest to study
conditions stronger than positivity, but weaker than complete
positivity. Such a condition is called {\it Kadison-Schwarz (KS)
property}, i.e a map $\phi$ satisfies the KS property
if $\phi(a)^*\phi(a)\leq \phi(a^*a)$ holds for every $a$. Note that
every unital completely positive map satisfies this inequality, and
a famous result of Kadison states that any positive unital map
satisfies the inequality for self-adjoint elements $a$. But KS-operators no need to be 
completely positive. In
\cite{Rob} relations between $n$-positivity of a map $\phi$ and the
KS property of certain map is established. Some nice
properties of the Kadison-Schwarz maps were investigated in
\cite{LMM,Rob1}.

In this paper we are going to describe KS-operators which are unital, trace preserving linear mappings (i.e. bistochastic operators) defined on the algebra of 2 by 2 matrices $\textit{M}_2(\mathbb{C})$. In Section 2 we show that the set of KS-operators forms a convex set. In section 3, we characterize bistochastic KS-operators on $M_2(\bc)$.  Such a description allows us to find
positive, but not Kadison-Schwarz operators. Moreover, by means of that characterization one can construct  KS-operators, which are not completely positive. Note that trace-preserving maps arise
naturally in quantum information theory \cite{KR,K,NC,RSW}
and other situations in which one wishes to restrict attention to a quantum system
that should properly be considered a subsystem of a larger system with which it interacts.

\section{Preliminaries}

Let $A$ and $B$ be unital $C^*$-algebras with identity $\id$. Recall that a linear mapping  $\Phi: A\to B$ is called {
\begin{enumerate}
\item[(i)] {\it positive} if  $\Phi(x)\geq0$ whenever $x\geq0$;

\item[(ii)] {\it unital} if   $\Phi(\id)=\id$;

 \item[(iii)] {\it $n$-positive} if
the mapping $\Phi_n: M_n(A)\to M_n(B)$ defined by $\Phi_n(a_{ij})=(\Phi(a_{ij}))$ is positive. Here $M_n(A)$ denotes the algebra of $n \times n$ matrices with $A$-valued entries;

\item[(iv)] {\it completely positive} if it is $n$-positive for all $n\in\bn$;

\item[(v)] {\it Kadison-Schwarz operator (KS-operator)}, if one has
\begin{eqnarray}\label{ks2}
\Phi(x)^*\Phi(x)\leq \Phi(x^*x) \ \ \textrm{for all} \ \ x\in A.
\end{eqnarray}
\end{enumerate}

It is clear that any KS-operator is positive.
Note that
every 2- positive map is KS-map, and
a famous result of Kadison states that any positive unital map
satisfies the inequality \eqref{ks2} for all self-adjoint elements $x\in A$.

By $\mathcal{KS}(A,B)$ we denote the set of all KS-operators mapping from $A$ to $B$.

\begin{thm}\label{ks-s} The following assertions hold true:
\begin{enumerate}
\item[(i)] Let $\Phi,\Psi\in \mathcal{KS}(A,B)$, then for any $\l\in[0,1]$ the mapping
$\Gamma=\l \Phi+(1-\l)\Psi$ belongs to $\mathcal{KS}(A_1,A_2)$. This means $\mathcal{KS}(A,B)$ is convex;
\item[(ii)] Let $U,V$ be unitaries in $A$ and $B$, respectively, then for any $\Phi\in \mathcal{KS}(A,B)$ the mapping $\Psi_{U,V}(x)=U\Phi(VxV^*)U^*$ belongs to $\mathcal{KS}(A,B)$.
\end{enumerate}
\end{thm}

\begin{proof} (i). Let us show that $\Gamma_\l$ satisfies \eqref{ks2}. Let $x\in A$, then one can see that
\begin{eqnarray}\label{vp1}
\Gamma_\l(x^*x)&=&\l \Phi(x^*x)+(1-\l)\Psi(x^*x)\nonumber\\
&\geq& \l \Phi(x)^*\Phi(x)+(1-\l)\Psi(x)^*\Psi(x)
\end{eqnarray}
and
\begin{eqnarray}\label{vp2}
\Gamma_\l(x)^*\Gamma_\l(x)&=&\l^2\Phi(x)^*\Phi(x)+\l(1-\l)\Phi(x)^*\Psi(x)\nonumber\\
&+&\l(1-\l)\Psi(x)^*\Phi(x)+(1-\l)^2\Psi(x)^*\Psi(x)
\end{eqnarray}
Hence, from  \eqref{vp1} - \eqref{vp2} one gets
\begin{eqnarray*}
\Gamma_\l(x^*x)-\Gamma_\l(x)^*\Gamma_\l(x)\geq \l(1-\l)\big(\Phi(x)-\Psi(x)\big)^*(\Phi(x)-\Psi(x))\geq 0,
\end{eqnarray*}
which proves the assertion.\\

(ii) For any $x\in A$ one has
\begin{eqnarray*}
\Psi_{U,V}(x^*x)&=& U\Phi\big((VxV^*)^*VxV^*\big)U^*\\
&\geq& U\Phi(VxV^*)^*\Phi(VxV^*\big)U^*\\
&=&  U\Phi(VxV^*)^*U^*U\Phi(VxV^*\big)U^*\\
&=&\Psi_{U,V}(x)^*\Psi_{U,V}(x),
\end{eqnarray*}
this completes the proof.
\end{proof}

Let us consider the set of 2 by 2 matrices $M_2(\bc)$ over $\bc$. In the
sequel by $\id$ we mean an identity matrix. By $\tr$ we mean trace on $M_2(\bc)$. In what follows by $\tau$ we denote a normalized trace, i.e. $\t=\frac12\tr$.

A linear mapping  $\Phi: M_2(\mathbb{C})\to M_2(\mathbb{C})$ is called {\it bistochastic} if it is positive, unital and trace preserving, i.e. $\tau(\Phi(x))=\tau(x)$ for all $x\in M_2(\mathbb{C})$).
Note that this terminology for maps that are both unital and
stochastic was introduced in \cite{AHW}.

In the paper we are going to consider bistochastic KS-operators on $M_2(\bc)$. Therefore, by $\mathcal{KS}(M_2(\bc))$ we denote the set of all bistochastic KS-operators defined on $M_2(\bc)$. According to Theorem \ref{ks-s} the set $\mathcal{KS}(M_2(\bc))$ is convex.

\section{Kadison-Schwarz operators on $M_2(\bc)$}

It is known (see \cite{BR,KR}) that the identity and the Pauli matrices $\{\id,\sigma_1,\s_2,\s_3\}$ form a basis for $M_2(\mathbb{C})$, where
\begin{equation*}
\s_1=\left(
      \begin{array}{cc}
        0 & 1 \\
        1 & 0 \\
      \end{array}
    \right) \
\s_2=\left(
       \begin{array}{cc}
         0 & -i \\
         i & 0 \\
       \end{array}
     \right) \
\s_3=\left(
       \begin{array}{cc}
         1 & 0 \\
         0 & -1 \\
       \end{array}
     \right).
\end{equation*}

Every matrix $a\in\textit{M}_2(\mathbb{C})$ can be written in this basis as $a=w_0\id+\wb\cdot\s$ with $w_0\in\mathbb{C}, \wb=(w_1,w_2,w_3)\in\mathbb{C}^3,$ here by $\wb\cdot\s$ we mean the following $$\wb\cdot\s=w_1\s_1+w_2\s_2+w_3\s_3.$$
The following facts holds (see \cite{RSW}):
\begin{enumerate}
\item[(a)] \ a matrix $a\in\textit{M}_2(\mathbb{C})$ is self-adjoint if and only if $w_0$ and $\wb$ are real;
\item[(b)] \ a matrix $a\in\textit{M}_2(\mathbb{C})$ is positive if and only if $\|\wb\|\leq w_0$, where
$$\|\wb\|=\sqrt{|w_1|^2+|w_2|^2+|w_3|^2} \ ;$$
\item[(c)] \ a matrix $a\in\textit{M}_2(\mathbb{C})$ is normal if and only if $[\wb,\overline{\wb}]=[\overline{\wb},\wb]$\\
for every $\wb\in\mathbb{C}^3$, where $[\cdot,\cdot]$ stands for the cross product of vectors in $\mathbb{C}^3$.
\end{enumerate}

Every $\Phi: M_2(\mathbb{C})\rightarrow M_2(\mathbb{C})$ linear mapping can also be represented in this basis by a unique $4\times4$ matrix $\textbf{F}$. It is trace preserving if and only if
$\textbf{F}=\left(
                                                                                    \begin{array}{cc}
                                                                                      1 & 0 \\
                                                                                      t & T \\
                                                                                    \end{array}
                                                                                  \right)$
where T is a $3\times3$ matrix and $0$ and $t$ are row and column vectors respectively so that
\begin{equation}\label{ks3}
\Phi(w_0\id+\wb\cdot\s)=w_0\id+(w_0t+T\wb)\cdot\s.
\end{equation}

When $\Phi$ is also positive then it maps the subspace of self-adjoint matrices of $M_2(\mathbb{C})$ into itself, which implies that $T$ is real. A linear mapping $\Phi$ is unital  if and only if $t=0$. So, in this case we have
\begin{equation}\label{ks4}
\Phi(w_0\id+\wb\cdot\s)=w_0\id+(T\wb)\cdot\s.
\end{equation}

Hence, any bistochastic mapping $\Phi:M_2(\mathbb{C})\rightarrow M_2(\mathbb{C})$ has a form \eqref{ks4}. Now we are going to give a characterization bistochastic KS-maps.

\begin{thm}\label{ks-d}
Any bistochastic mapping $\Phi: {M}_2(\mathbb{C})\rightarrow M_2(\mathbb{C})$ is  KS-operator if and only if one has
\begin{eqnarray}\label{ks5}
&&\|T\wb\|\leq\|\wb\|, \ \ T\overline{\wb}=\overline{T\wb} \\ \label{ks6}
&&\bigg\|T[\wb,\overline{\wb}]-\big[T\wb,\overline{T\wb}\big]\bigg\|\leq\|\wb\|^2-\|T\wb\|^2
\end{eqnarray}
for all $\wb\in \bc^3$.
\end{thm}
\begin{proof}
'if' part. Let $x\in\textit{M}_2(\mathbb{C})$ be an arbitrary element, i.e. $x=w_0\id+\wb\cdot\s.$ Then $x^*=\overline{w_0}\id+\overline{\wb}\cdot\s$. Therefore
\begin{equation*}
x^*x=\big(|w_0|^2+\|\wb\|^2\big)\id+\big(w_0\overline{\wb}+\overline{w_0}\wb-i\big[\wb,\overline{\wb}\big]\big)\cdot\s
\end{equation*}
Consequently, we have
\begin{eqnarray}\label{ks7}
\Phi(x)=w_0\id+(T\wb)\cdot\s, \ \  \Phi(x^*)=\overline{w_0}\id+\big(T\overline{\wb}\big)\cdot\s \\ \label{ks8}
\Phi(x^*x)=\big(|w_0|^2+\|\wb\|^2\big)\id+\big(w_0T\overline{\wb}+\overline{w_0}T\wb-iT\big[\wb,\overline{\wb}\big]\big)\cdot\s\\ \label{ks9}
\Phi(x)^*\Phi(x)=\big(|w_0|^2+\|T\wb\|^2\big)\id+\big(w_0\overline{T\wb}+\overline{w_0}T\wb-i\big[T\wb,\overline{T\wb}\big]\big)\cdot\s
\end{eqnarray}
From \eqref{ks8}-\eqref{ks9} one gets
\begin{eqnarray*}
\Phi(x^*x)-\Phi(x)^*\Phi(x)&=&\big(|w_0|^2-\|T\wb\|^2\big)\id\\
&&+\bigg(w_0\big(T\overline{\wb}-\overline{T\wb}\big)-i\big(T\big[\wb,\overline{\wb}\big]-\big[T\wb,\overline{T\wb}\big]\big)\bigg)\cdot\s\geq0
\end{eqnarray*}
Hence, due to (b) we conclude that$\Phi$ should be positive, which means $T$ is real, therefore one gets $T\overline{\wb}=\overline{T\wb}$.
consequently, the last inequality yields
\begin{eqnarray}\label{ks10}
\big(|w_0|^2-\|T\wb\|^2\big)\id-i\big(T\big[\wb,\overline{\wb}\big]-\big[T\wb,\overline{T\wb}\big]\big)\cdot\s\geq0
\end{eqnarray}
which again with (b) implies the assertion.

'only if' part. Let \eqref{ks5}-\eqref{ks6} be satisfied. Then we have \eqref{ks10}, which with \eqref{ks6} and \eqref{ks8}-\eqref{ks9} yields
\eqref{ks2}. This completes the proof.
\end{proof}

Let $\Phi$ be a bistochastic KS-operator on $M_2(\bc)$, then it can be represented by \eqref{ks4}. Following \cite{KR} let us decompose the matrix $T$ as follows $T=RS$, here $R$ is a rotation and $S$ is a self-adjoint matrix (see \cite{KR}). Define a mapping $\Phi_S$ as follows
\begin{equation}\label{F-s}
\Phi_S(w_0\id+\wb\cdot\s)=w_0\id+(S\wb)\cdot\s.
\end{equation}
Every rotation is implemented by a unitary matrix in $M_2(\bc)$, therefore there is a unitary $U\in M_2(\bc)$ such that
\begin{equation}\label{F-sU}
\Phi(x)=U\Phi_S(x)U^*, \ \ \ x\in M_2(\bc).
\end{equation}

On the other hand, every self-adjoint operator $S$ can be
diagonalized by some unitary operator, i.e. there is a unitary $V\in M_2(\bc)$ such that
$S=VD_{\l_1,\l_2,\l_3}V^*$, where
\begin{eqnarray}\label{ks-D}
D_{\l_1,\l_2,\l_3}=\left(
    \begin{array}{ccc}
      \l_1 & 0 & 0 \\
      0 & \l_2 & 0 \\
      0 & 0 & \l_3 \\
    \end{array}
  \right),
\end{eqnarray}
where $\l_1,\l_2,\l_3\in\br$.

Consequently, the mapping $\Phi$ can be represented by
\begin{equation}\label{F-DU}
\Phi(x)=\tilde U\Phi_{D_{\l_1,\l_2,\l_3}}(x)\tilde U^*, \ \ \ x\in M_2(\bc)
\end{equation}
for some unitary $\tilde U$. Due to Theorem \ref{ks-s} the mapping $\Phi_{D_{\l_1,\l_2,\l_3}}$ is also KS-operator.
Hence, all bistochastic KS-operators can be characterized by $\Phi_{D_{\l_1,\l_2,\l_3}}$ and unitaries.
In what follows, for the sake of shortness by $\Phi_{(\lambda_1,\lambda_2,\lambda_3)}$ we denote
the mapping $\Phi_{D_{\l_1,\l_2,\l_3}}$. It is clear to observe from \eqref{ks5} that  $|\lambda_k|\leq1, k=1,2,3$.
It is easy to see that the mapping $\Phi_D\mapsto U\Phi_DU^*$ is affine, therefore, if $\Phi_D$ is an extreme point of
${\mathcal{KS}}(M_2(\bc))$ then $U\Phi_DU^*$ is an extreme point of $\mathcal{KS}(M_2(\bc)$ as well. Denote
\begin{equation}\label{l-ks}
\D=\big\{(\lambda_1,\lambda_2,\lambda_3)\in\br^3:\ \ \Phi_{(\lambda_1,\lambda_2,\lambda_3)}\in {\mathcal{KS}}(M_2(\bc))\big\}.
\end{equation}
According to Theorem \ref{ks-s} the set $\D$ is convex. Now taking into account that the mapping $(\lambda_1,\lambda_2,\lambda_3)\mapsto \Phi_{(\lambda_1,\lambda_2,\lambda_3)}$ is affine, we infer that if 
$(\lambda_1,\lambda_2,\lambda_3)$ is an extreme point of $\D$, then $\Phi_{(\lambda_1,\lambda_2,\lambda_3)}$ is also
extreme point of ${\mathcal{KS}}(M_2(\bc))$. \\

{\bf Example 1.} Let us consider a famous example of non completely positive operator defined by transposition, i.e. $\Phi(x)=x^T$, where for $x\in M_2(\bc)$ by $x^T$ we denote its transposition. This mapping can be written in terms of $\Phi_{(\lambda_1,\lambda_2,\lambda_3)}$ as follows  $\Phi=\Phi_{(1,-1,1)}$.  First observe that
by taking $\wb=(1,1,i)$ in \eqref{ks6} one finds
\begin{eqnarray*}
2\sqrt{(\lambda_1-\lambda_2\lambda_3)^2+(\lambda_2-\lambda_1\lambda_3)^2}\leq2-\lambda_1^2-\lambda_2^2+1-\lambda_3^2.
\end{eqnarray*}
Putting $\lambda_3=1$, then the last one can be written as follows
\begin{eqnarray}\label{ks11}
2\sqrt2|\lambda_1-\lambda_2|\leq2-\lambda_1^2-\lambda_2^2.
\end{eqnarray}
It is clear that at $\lambda_1=1$, $\lambda_2=-1$, the inequality \eqref{ks11} is not satisfied, hence \eqref{ks6} as well. This means that $\Phi_{(1,-1,1)}$ is positive, but not KS-map.\\

In \cite{RSW} it has been given a characterization of completely positivity of $\Phi_{(\lambda_1,\lambda_2,\lambda_3)}$. Namely, the following result holds.

\begin{thm}\label{CP-m} A map $\Phi_{(\lambda_1,\lambda_2,\lambda_3)}$ is complete positive if and only if the followings inequalities are satisfied
\begin{eqnarray}\label{ks14}
&&(\l_1+\l_2)^2\leq(1+\l_3)^2\\ \label{ks15}
&&(\l_1-\l_2)^2\leq(1-\l_3)^2\\ \label{ks16}
&&(1-(\l_1^2+\l_2^2+\l_3^2))^2\geq 4(\l_1^2\l_2^2+\l_2^2\l_3^2+\l_1^2\l_3^2-2\l_1\l_2\l_3)
\end{eqnarray}
\end{thm}

Let us characterize KS operators of the form $\Phi_{(\lambda_1,\lambda_2,\lambda_3)}$.

Using simple calculation from \eqref{ks6} with $T=D_{\lambda_1,\lambda_2,\lambda_3}$ we obtain the following
\begin{eqnarray}\label{34-1}
A|w_2\overline{w}_3-w_3\overline{w}_2|^2&+&B|w_1\overline{w}_3-w_3\overline{w}_1|^2\nonumber \\[2mm]
&&+C|w_1\overline{w}_2-w_2\overline{w}_1|^2\leq \big(\a|w_1|^2+\b |w_2|^2+\g |w_3|^2\big)^2,
\end{eqnarray}
where $\wb=(w_1,w_2,w_3)\in \bc^3$ and
\begin{eqnarray}\label{abc}
&&\alpha=|1-\lambda_1^2|,  \ \ \ \beta=|1-\lambda_2^2|, \ \ \ \gamma=|1-\lambda_3^2|\\
\label{abc1}
&&A=|\lambda_1-\lambda_2\lambda_3|^2, \ \ \ B=|\lambda_2-\lambda_1\lambda_3|^2, \ \ \ C=|\lambda_3-\lambda_1\lambda_2|^2.
\end{eqnarray}

Due to the inequality $|2\Re(uv)|\leq |u|^2+|v|^2$, we have
\begin{eqnarray*}
|w_i\overline{w}_j-w_j\overline{w}_i|^2=|2\Re(w_iw_j)|^2\leq |w_i|^4+2|w_i|^2|w_j|^2+|w_j|^4 \ \ (i\neq j)
\end{eqnarray*}
hence, we estimate LHS of \eqref{34-1} by
\begin{eqnarray*}
A(|w_2|^4+2|w_2|^2|w_3|^2+|w_3|^4)+B(|w_1|^4+2|w_1|^2|w_3|^2+|w_3|^4)-C(|w_1|^4+2|w_1|^2|w_2|^2+|w_2|^4)
\end{eqnarray*}
Consequently, from \eqref{34-1} we derive the following one
\begin{eqnarray}
&&|w_1|^4(\alpha^2-B-C)+|w_2|^4(\beta^2-A-C)+|w_3|^4(\gamma^2-A-B)\nonumber \\
\label{ks12}
&&+2|w_1|^2|w_2|^2(\alpha\beta-C)+2|w_1|^2|w_3|^2(\alpha\gamma-B)+2|w_2|^2|w_3|^2(\beta\gamma-A)\geq0
\end{eqnarray}

It is easy to see that \eqref{ks12} is satisfied if one has
\begin{eqnarray*}
&&\a^2\geq B+C, \quad  \b^2\geq A+C, \quad  \g^2\geq A+B,\\
&&\a\b\geq C,  \quad  \a\g\geq B,   \quad  \b\g\geq A.
\end{eqnarray*}

Substituting above denotations \eqref{abc},\eqref{abc1} to the last inequalities, and doing simple
 calculation one derives

\begin{eqnarray}\label{ks17}
&&(1+\l_1^2)(3+\l_2^2+\l_3^2-\l_1^2)\leq 4(1+\l_1\l_2\l_3);\\ \label{ks18}
&&(1+\l_2^2)(3+\l_1^2+\l_3^2-\l_2^2)\leq 4(1+\l_1\l_2\l_3);\\ \label{ks19}
&&(1+\l_3^2)(3+\l_1^2+\l_2^2-\l_3^2)\leq 4(1+\l_1\l_2\l_3);\\ \label{ks20}
&&\l_1^2+\l_2^2+\l_3^2\leq 1+ 2\l_1\l_2\l_3.
\end{eqnarray}

Hence, we have the following

\begin{thm}\label{KS-m}
If \eqref{ks17},\eqref{ks18},\eqref{ks19} and \eqref{ks20} are satisfied, then a map $\Phi_{(\lambda_1,\lambda_2,\lambda_3)}$ is a KS-operator.
\end{thm}

The last theorem allows us to construct lots of KS-operators, which are not completely positive (see Example 2).

Let us consider mappings $\Phi_{(\l,\l,\m)}$, and for the check the conditions of Theorems \ref{CP-m} and
\ref{KS-m}.

Calculating \eqref{ks14}, \eqref{ks15}, \eqref{ks16}, \eqref{ks17},\eqref{ks18},\eqref{ks19} and \eqref{ks20} we obtain the following
\begin{eqnarray*}
\l^2\leq\frac{(1+\m)^2}{4}; \ \ \ \ \ \l^2\leq\frac{1+\m}{3-\m}; \ \ \ \ \ \l^2\leq\frac{(1+\m)^2}2; \ \ \ \ \ \l^2\leq\frac{1+\m}2.
\end{eqnarray*}

The graphics of above inequalities are the following
\begin{figure}[h]
\centering
\includegraphics[width=0.5\columnwidth,clip]{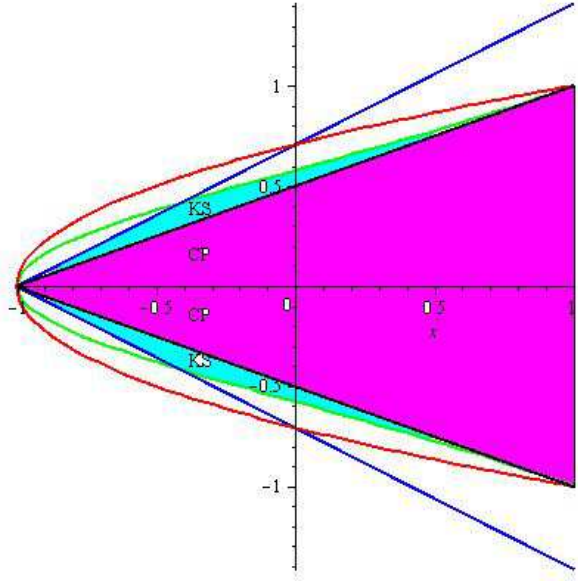}\\
\caption{Blue color indicates KS operators, which are not CP. Red color indicated CP maps.}
\end{figure}

From the graphic we see that class of KS-operators are much larger that the class of completely positive ones.

{\bf Example 2.} Now we are going to construct KS-operators, which is not complete positive.
Consider mappings of the form  $\Phi_{(\l,\l,\l)}$. Let us first check conditions of Theorem \ref{CP-m}, here as above
$|\l|\leq1$. From \eqref{ks14} we obtain the following inequality $$4\l^2\leq(1+\l)^2$$
Solving the last inequality one has $$(\l-1)(\l+\frac13)\leq0$$
If we take $\l$ such that $-1\leq\l<-\frac13$, then \eqref{ks14} is not satisfied. This means that $\Phi_{(\l,\l,\l)}$ is not complete positive.

Next we are going to check conditions of Theorem \ref{KS-m} From \eqref{ks17},\eqref{ks18},\eqref{ks19} and \eqref{ks20} one finds $$(1+\l^2)(3+\l^2)\leq4(1+\l^3).$$
Calculating the last one we obtain $$(\l-1)^2(\l-1-\sqrt2)(\l-1+\sqrt2)\leq0.$$
If $1-\sqrt2\leq\l\leq1$ then $\Phi_{(\l,\l,\l)}$ is KS-operator.

So, taking into account above we conclude that  if $1-\sqrt2\leq\l<-\frac13$, then $\Phi_{(\l,\l,\l)}$ is KS-operator,
 but not complete positive one.

\section*{Acknowledgement} This work is partially supported by
the Malaysian Ministry of Science, Technology and Innovation Grant
01-01-08-SF0079. A finial part of
this work was done at the Abdus Salam International
Center for Theoretical Physics (ICTP), Trieste, Italy. F.M. thanks
the ICTP for providing financial support of his visit (within the scheme of
Junior Associate) to ICTP.


\begin{thebibliography}{92}

\bibitem{AHW} G.G. Amosov, A.S. Holevo, and R.F. Werner, On the additivity hypothesis in quantum information theory,  {\it Probl. Inf. Transm.}  {\bf 36}  (2000),  305--313

\bibitem{BR} O. Bratteli and D. W. Robertson, Operator algebras and quantum
statistical mechanics. {\it I, Springer, New York–Heidelberg–Berlin}
(1979).

\bibitem{Choi} M-D. Choi, Completely positive linear maps on complex matrices, {\it Lin. Alg. Appl.} {\bf 10}(1975),
285--290.

\bibitem{KR} C. King, M.B. Ruskai, Minimal entropy of states emerging from noisy quantum
channels. {\it IEEE Trans. Info. Theory} {\bf 47}, (2001) 192–-209 .

\bibitem{K} A. Kossakowski, A Class of Linear Positive Maps in Matrix
Algebras, {\it Open Sys. \& Information Dyn.} {\bf 10}(2003)
213--220.

\bibitem{LMM}  L.E. Labuschagne, W.A. Majewski, M. Marciniak, On $k$-decomposability of positive maps.  {\it Expo. Math.} {\bf 24} (2006),  103--125.

\bibitem{MM1} W.A. Majewski, M. Marciniak,  On a characterization of positive maps, {\it J. Phys. A: Math. Gen.} {\bf 34} (2001) 5863-–5874.


\bibitem{Ma1} W.A. Majewski, On non-completely positive quantum dynamical maps on spin chains, {\it J. Phys. A: Math. Gen.} {\bf 40} (2007) 11539-–11545.

\bibitem{Ma2} W.A. Majewski, On positive decomposable maps.  {\it Rep. Math. Phys.} {\bf  59} (2007), 289--298.
    
\bibitem{NC} M.A. Nielsen, I.L. Chuang, {\it Quantum Computation and Quantum Information}, 
Cambridge University Press,  2000.

\bibitem{RSW} M.B. Ruskai, S. Szarek, E. Werner, An analysis of completely positive
trace-preserving maps on $M_2$, {\it Lin. Alg. Appl.} {\bf 347}
(2002) 159–-187.




\bibitem{Rob1} A. G. Robertson, A Korovkin theorem for Schwarz maps on $C^*$-algebras, {\it Math. Z.} {\bf 156}(1977),
205--206.

\bibitem{Rob} A. G. Robertson, Schwarz inequalities and the decomposition of positive maps on $C^*$-algebras {\it Math. Proc. Camb. Philos. Soc.} {\bf 94}(1983),
291--296.

\bibitem{St} E. Stormer, Positive linear maps of operator algebras, {\it Acta Math.} {\bf 110}(1963), 233--278.
\bibitem{W} S. L. Woronowicz, Positive maps of low dimensional matrix algebras, {\it Rep. Math. Phys.}
 {\bf 10} (1976), 165--183.

\end{thebibliography}
\end{document}